\documentclass{article}
\newcommand\version{0}

\usepackage[margin=3cm]{geometry}
\usepackage{appendix}
\usepackage{mathtools}
\usepackage{amsmath}
\usepackage{amsthm}
\usepackage{amssymb}
\usepackage{amsfonts}
\usepackage{bbm}
\usepackage[svgnames]{xcolor}
\usepackage{tikz}
\usepackage{verbatim}
\usepackage{hyperref}
\usepackage{fullpage}
\usepackage{enumitem}
\usepackage{subcaption}
\usepackage{makecell}
\usepackage{mathdots}
\usepackage{etoolbox}
\usepackage{standalone}
\usepackage{pgfplots}
\usepackage[font={small,it}]{caption}

\usetikzlibrary{decorations.markings}
\usetikzlibrary{shapes.geometric}
\pgfdeclarelayer{edgelayer}
\pgfdeclarelayer{nodelayer}
\pgfsetlayers{edgelayer,nodelayer,main}
\tikzstyle{none}=[inner sep=0pt]
\tikzstyle{rn}=[circle,fill=Red,draw=Black,line width=0.8 pt]
\tikzstyle{gn}=[circle,fill=Lime,draw=Black,line width=0.8 pt]
\tikzstyle{yn}=[circle,fill=Yellow,draw=Black,line width=0.8 pt]
\tikzstyle{simple}=[-,draw=Black,line width=2.000]
\tikzstyle{arrow}=[-,draw=Black,postaction={decorate},decoration={markings,mark=at position .5 with {\arrow{>}}},line width=2.000]
\tikzstyle{tick}=[-,draw=Black,postaction={decorate},decoration={markings,mark=at position .5 with {\draw (0,-0.1) -- (0,0.1);}},line width=2.000]

\pgfdeclarelayer{edgelayer}
\pgfdeclarelayer{nodelayer}
\pgfsetlayers{edgelayer,nodelayer,main}

\tikzstyle{none}=[inner sep=0pt]

\tikzstyle{rn}=[circle,fill=Red,draw=Black,line width=0.8 pt]
\tikzstyle{gn}=[circle,fill=Lime,draw=Black,line width=0.8 pt]
\tikzstyle{yn}=[circle,fill=Yellow,draw=Black,line width=0.8 pt]

\tikzstyle{simple}=[-,draw=Black,line width=2.000]
\tikzstyle{arrow}=[-,draw=Black,postaction={decorate},decoration={markings,mark=at position .5 with {\arrow{>}}},line width=2.000]
\tikzstyle{tick}=[-,draw=Black,postaction={decorate},decoration={markings,mark=at position .5 with {\draw (0,-0.1) -- (0,0.1);}},line width=2.000]

\theoremstyle{plain}
\newtheorem{theorem}{Theorem}[section]
\newtheorem{proposition}[theorem]{Proposition}

\newtheorem{lemma}[theorem]{Lemma}

\theoremstyle{definition}
\newtheorem{procedure}[theorem]{Procedure}
\newtheorem{definition}[theorem]{Definition}
\newtheorem{remark}[theorem]{Remark}
\newtheorem{example}[theorem]{Example}

\def\Tr{\ensuremath{\mathrm{Tr}}}

\def\Re{\ensuremath{\mathrm{Re}}}
\def\T{\ensuremath{\mathrm{T}}}
\def\Ker{\ensuremath{\mathrm{Ker}}}

\newcounter{jamiecomment}
\newcounter{dominiccomment}
\newcommand\DVcomm[1]{\ensuremath{{}^{\color{blue}\thedominiccomment}}\marginpar{\color{blue}\tiny\raggedright \thedominiccomment: #1}\stepcounter{dominiccomment}}
\newcommand\ignore[1]{}
\def\C{\mathbb{C}}
\def\Z{\mathbb{Z}}
\def\d{\mathrm{d}}

\newcommand{\bra}[1]{\left\langle #1 \right|\,}
\newcommand{\ket}[1]{\,\left| #1 \right\rangle}
\newcommand{\braket}[2]{\left.\left\langle {#1} \,\right|\,{#2} \right\rangle }

\DeclarePairedDelimiter\floor{\lfloor}{\rfloor}
\DeclarePairedDelimiter\abs{|}{|}
\def\SU{\ensuremath{\mathrm{SU}}}
\def\U{\ensuremath{\mathrm{U}}}
\def\SO{\ensuremath{\mathrm{SO}}}

\def\BOct{\ensuremath{\mathrm{BOct}}}
\def\BTet{\ensuremath{\mathrm{BTet}}}

\def\Tet{\ensuremath{\mathrm{Tet}}}
\newcommand{\infiniteabstract}{\begin{abstract}
We present two new schemes for quantum teleportation between parties whose local reference frames are misaligned by the action of a compact Lie group $G$. These schemes require no prior alignment of reference frames and are unaffected by arbitrary changes in reference frame alignment during execution, suiting them to situations of rapid reference frame drift. Our \textit{tight} scheme yields improved purity compared to standard teleportation, in some cases substantially --- this includes the case of qubit teleportation under arbitrary $\SU(2)$ reference frame uncertainty--- while communicating no information about either party's reference frame alignment at any time. Our \emph{perfect} scheme performs perfect teleportation, but does communicate some reference frame information. The mathematical foundation of these schemes is a unitary error basis permuted up to a phase by the conjugation action of a finite subgroup of $G$.
\end{abstract}}

\newcommand{\infiniteacknowledgements}{The authors thank Niel de Beaudrap, Matty Hoban,  Carlo-Maria Scandolo and Nathan Walk for useful discussions regarding measures of channel quality. They are also grateful to Jean-Philippe Bourgoin, Matthias Fink, Reiner Kaltenbaek and all others who shared their expertise at the 
Lisbon Training Workshop on Quantum Technologies in Space. The first author acknowledges support from the Engineering and Physical Sciences Research Council.
}

\newcommand{\infinitetitleetc}[1]{\ifnumcomp{#1}{=}{1}{\title{Quantum teleportation with infinite reference frame uncertainty}
\date{\today}
\author{Dominic Verdon}\email{dominic.verdon@bristol.ac.uk}
\affiliation{School of Mathematics, University of Bristol, University Walk, Bristol BS8 1TW}
\author{Jamie Vicary}\email{j.o.vicary@bham.ac.uk}
\affiliation{School of Computer Science, University of Birmingham, University Road\ West, Birmingham B15 2TT} \infiniteabstract \maketitle}
{
\title{\vspace{-2.5cm}\bf \LARGE \mbox{Quantum teleportation with infinite reference} \mbox{frame uncertainty} and without prior alignment}
\author{\vspace{-.70cm}
\\
Dominic Verdon\thanks{\texttt{dominic.verdon@cs.ox.ac.uk}} \ and Jamie Vicary\thanks{\texttt{jamie.vicary@cs.ox.ac.uk}}
\\
Department of Computer Science, University of Oxford\vspace{-5pt}}
\maketitle \vspace{-1cm} \infiniteabstract}}

\newcommand{\infiniteintroacknowledgements}[1]{\ifnumcomp{#1}{=}{1}{}{
\infiniteacknowledgements}}
\newcommand{\infinitepraacknowledgements}[1]
{\ifnumcomp{#1}{=}{1}{\acknowledgements{\infiniteacknowledgements}}{}}
\newcommand{\biblstyle}[1]{\ifnumcomp{#1}{=}{1}{}{\bibliographystyle{plainurl}}}

\begin{document}

\infinitetitleetc{\version}
\section{Overview}
\label{sec:intro}
\paragraph{Motivation.}

A shared reference frame is an important implicit assumption underlying the correct execution of many multi-party quantum protocols~\cite{Bartlett2007, Marvian2013, Marvian2014, Kitaev2004, Enk2001, Gottesman1999, Gisin2002}. As quantum technologies move into space~\cite{Ren2017, Yin2017, Bacsardi2013} and into handheld devices~\cite{Wabnig2013, Duligall2006, Duligall2007}, scenarios where this assumption is violated are naturally encountered. This problem has already received considerable attention in the case of ground-to-satellite quantum key distribution~\cite{Laing2010, Liang2014, Bacsardi2013}; there is also a smaller body of work on quantum teleportation without a shared reference frame~\cite{Chiribella2012, Marzolino2015, Marzolino2016}, a subject which is increasingly important as quantum repeaters~\cite{Muralidharan2016} and ground-to-satellite quantum teleportation~\cite{Ren2017} become experimentally viable.

Prior alignment of reference frames~\cite{Bartlett2007,Skotiniotis2012,Islam2014,Islam2016,Peres2002} may become impractical in the case of time-varying misalignment, or where the parties are far apart; prior alignment also involves communication of reference frame information, which may be cryptographically sensitive~\cite{Ioannou2014, Bartlett2004, Kitaev2004}. Another approach involves the use of decoherence-free subspaces~\cite{Lidar2003}; because this requires larger Hilbert spaces, practical implementation can be nontrivial, although experimental solutions have been developed for optical systems~\cite{DAmbrosio2012}.

\paragraph{Our approach.}
We use a classical channel whose configurations are interpreted with respect to the local reference frame, such as might be used for prior alignment. Indeed, such a channel could be used to align frames by observing how a pre-agreed configuration transmitted by Alice is perceived by Bob. However, this does \emph{not} occur in our schemes; in particular, our schemes work when rapidly-varying reference frame alignment renders prior alignment impossible, and our tight scheme in fact communicates no information about either party's frame configuration at any time. Rather, in our schemes, Alice communicates \emph{the measurement result itself} using this channel.  If the parties' frames are not aligned, Bob will perform correction operations with respect to his own frame; these may not correspond to the measurement Alice performed, causing error. In our approach, however, the misalignment also causes errors in transmission of the measurement result; Bob may receive a different index to that sent by Alice. These errors are correlated, and our key idea is to construct schemes where they cancel out.

\paragraph{Equivariant unitary error bases.} A standard teleportation protocol can be described mathematically in terms of a \textit{unitary error basis} (UEB)~\cite{Werner2001}, a basis of unitary operators on a Hilbert space $\C^d$ which are orthogonal under the trace inner product. Let $G$ be a finite reference frame transformation group; we define a UEB to be \textit{$G$-equivariant} when its elements are permuted up to a phase under conjugation by $\rho(g)$ for any $g \in G$, where $\rho:G \to U(d)$ is the representation of $G$ on Bob's system~\cite{Bartlett2007}.

Equivariant UEBs are the mathematical foundation of our teleportation schemes. In previous work we exhaustively classified these for qubit systems~\cite[Thm. 4.1]{Verdon2018}; they exist precisely when the image of the composite homomorphism $G \stackrel \rho \to \U(2) \stackrel \tau \to \SO(3)$ is isomorphic to 1, $\Z_2$, $\Z_3$, $\Z_4$, $D_2$, $D_3$, $D_4$, $A_4$ or $S_4$, where $\tau$ is the obvious projection. We also provided constructions in higher dimension, and a method for proving nonexistence in some cases.

\paragraph{Tight scheme.} For any finite subgroup $H \subseteq G$ admitting an $H$-equivariant UEB, we construct a \emph{tight} teleportation scheme immune to reference frame errors arising from $H$. When $H=G$, the protocol allows error-free teleportation. When $G$ is larger than $H$, the protocol roughly allows us to `quotient' by the subgroup $H$, restricting the error to a \emph{fundamental domain} for $H$ in $G$. (See Figure~\ref{fig:fundamentaldomainintro}.)
\begin{figure}\centering
\includegraphics[scale=1]{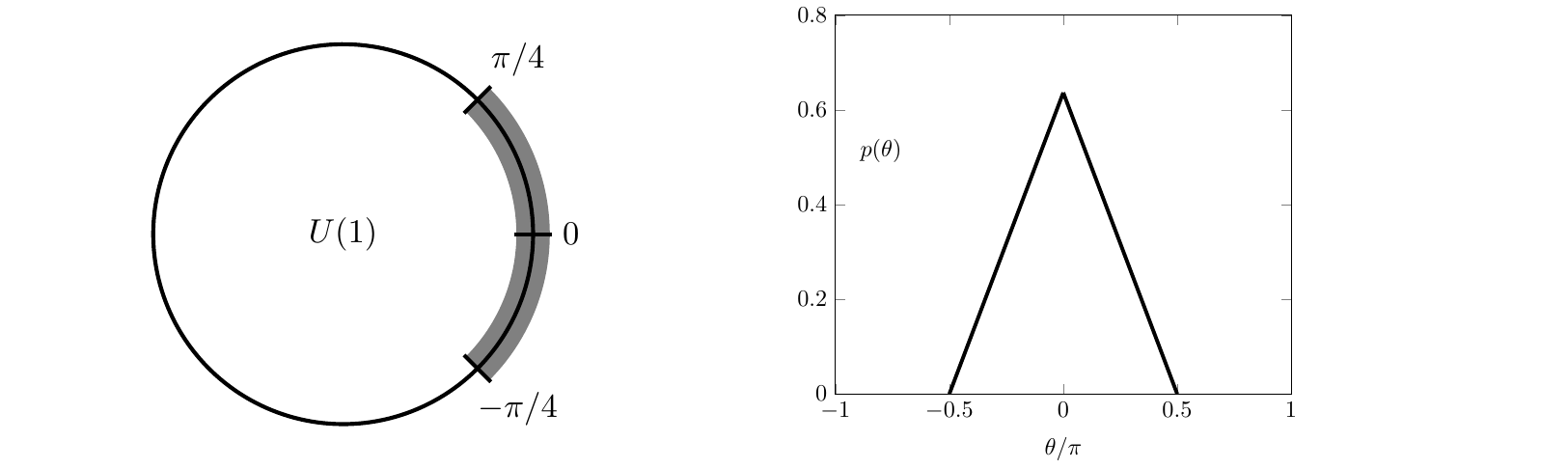}
\caption{\small The effective channel for a conventional protocol with uniform $\U(1)$ reference frame uncertainty is a uniform average over the channels induced by all misalignments $\theta \in [-\pi, \pi)$.  The cyclic subgroup $\mathbb{Z}_4 \subset U(1)$ possesses an equivariant UEB, allowing our tight scheme to `quotient out' $\mathbb{Z}_4$ reference frame uncertainty. Roughly, this reduces uncertainty to the region $\theta \in (-\pi/4,\pi/4)$ highlighted in the left subfigure; more precisely, the average over all misalignments is now weighted by $p(\theta)$, shown in the right subfigure.}
\label{fig:fundamentaldomainintro}
\end{figure}
 This can result in significant improvements in channel purity\footnote{We define this as the purity of the Choi-Jamio\l kowski state associated to the quantum channel induced by the protocol, where we take a convex sum over all frames $g \in G$ weighted by the Haar measure. Where figures are computed by numerical methods we give an error range in the reported figure.} compared to conventional teleportation, even for infinite compact Lie groups. For $G = \SU(2)$, for example, corresponding to arbitrary reference frame uncertainty for a qubit system, standard teleportation yields an average channel purity of $0.21$; with our tight scheme for the subgroup $\BOct \subset \SU(2)$, where $\BOct$ is the binary octahedral group, we obtain a channel purity of $0.44 \pm 0.03$, more than double that for standard teleportation. The results are shown in Table~\ref{tbl:numericsintro}. 
\begin{table}
\centering
\begin{tabular}{| c | c | c |}
Transformation group & Conventional purity & New tight scheme purity \\
\hline 
$\U(1)$ & 0.59 &        0.65 (matched channel) \\[6pt]
$\SU(2)$ & 0.21 & \parbox{5cm}{$0.32 \pm 0.02$ (matched channel) \\[-5pt] $0.44 \pm 0.03$ (rod channel)}
\end{tabular}
\caption{Qubit teleportation using a matched channel for $\U(1)$ and $\SU(2)$ reference frame uncertainty. The numbers shown are the purities of the effective quantum channels.}
\label{tbl:numericsintro}
\end{table}
The tight scheme additionally possesses the following desirable properties:
\begin{itemize}
\item \emph{Dynamical robustness (DR).} It is unaffected by arbitrary changes in reference frame alignment during transmission of the measurement result, provided Bob's frame alignment remains approximately constant between his receipt of the measurement result and his performance of the unitary correction.
\item \emph{Minimal entanglement (ME).} The parties only require a  $d$-dimensional maximally entangled resource state.
\item \emph{Minimal communication (MC).} Only $2$ dits of classical information are communicated from Alice to Bob.
\item \emph{No reference frame leakage (NL).} No information about either party's reference frame alignment at any time is communicated. (This property is of cryptographic significance~\cite{Ioannou2014, Bartlett2004, Kitaev2004}.)
\end{itemize}

\paragraph{Perfect scheme.} The tight scheme yields an improvement in the quality of the channel. Our perfect scheme, on the other hand, performs \emph{perfect} teleportation, up to a global phase, while retaining properties (DR) and (ME) and without communicating full information about Alice's frame configuration at the time of measurement.  To achieve this, additional reference frame information is transmitted by Alice in the same package as the measurement result, reducing reference frame uncertainty exactly to the finite group $H$, for which perfect teleportation is possible. Our techniques allow us to `fold' the measurement result in with the reference frame information, obviating the need to communicate it through a separate channel and, importantly, maintaining the novel (DR) property.

\ignore{
\paragraph{Worked examples.}We obtain a full expression for the channel induced by the tight scheme in general, and present the following examples in full detail:
\begin{itemize}
\item Phase reference frame uncertainty $G=\U(1)$ (\cite[Section 6]{VerdonInfinite}). Alice and Bob share an optical link; systems are photon polarisation states and the reference frame is a 2\-dimensional Cartesian frame determining the $x$- and $y$-polarisation axes. Alice encodes the measurement result in the linear polarisation direction of a pulse of classical light.
\item Spatial reference frame uncertainty  $G=\SU(2)$ (\cite[Section 2]{VerdonInfinite}). Systems are spin\-$\frac 1 2$ particles and the reference frame is a 3-dimensional Cartesian frame determining the spin axes. Alice encodes the measurement result in the orientation of a direction-bearing classical object sent to Bob by parallel transport.
\end{itemize}
}

\paragraph{Related work.}
Chiribella et al~\cite{Chiribella2012} argued that, when the reference transformation group is a  continuous compact Lie group, there is no teleportation procedure yielding perfect state transfer. They did not consider transmission of the measurement result in a reference frame--dependent manner, and their no-go theorem therefore does not apply to our results.

Some other approaches for finite $G$ can be found in the literature. These rely on a variety of techniques: using additional pre-shared entanglement~\cite{Chiribella2012}; sharing additional entanglement during the protocol~\cite{Kitaev2004}; and transmitting more complex resources~\cite[Section V.A]{Bartlett2007}.
None of these share the (DR) property, and they all require additional resources and additional quantum operations.

\paragraph{Outlook.} 
Work has been done on reference frame--independent quantum key distribution between handheld devices sharing an optical link~\cite{Wabnig2013, Duligall2006, Duligall2007}; such devices seem an obvious application for our perfect scheme for $\U(1)$ uncertainty. 
\ignore{We believe that our overall approach may be extended to perform other multi-party protocols, including quantum key  distribution~\cite{Bennett2014,Ekert1991}, in the case of rapidly varying frame alignment, and will pursue this in future work.}
There may also be cryptographic applications for these results, as it has been noted that a private shared reference frame may be used as a secret key~\cite{Kitaev2004, Ioannou2014, Bartlett2004}, and our tight scheme does not leak reference frame information.

\ignore{

\paragraph{Outline.} 
In Section~\ref{sec:example} we give an informal example of our approach in the case of spatial reference frame uncertainty. In Section~\ref{sec:theorysection} we make more precise definitions, define the Type A and B teleportation procedure, and derive an expression for the induced teleportation channel. In Section~\ref{sec:encodings} we define Type A and Type B encoding schemes. In Section~\ref{sec:referenceframeencodings} we give a general construction for compatible Type A and B encoding schemes on matched channels and characterise when Type C teleportation schemes are possible. In Section~\ref{sec:example2} we apply our reference frame constructions to a realistic example of phase reference frame uncertainty. In Section~\ref{sec:numerics} we derive the numerical results in Table~\ref{tbl:numericsintro}. In the appendices we prove some results about the effect of reference frame changes on states and operations, the main theorem from Section~\ref{sec:theorysection}, and some facts about Voronoi cells used in Section~\ref{sec:numerics}.}

\ignore{\subsection{Acknowledgements}}

\infiniteintroacknowledgements{\version}

\section{Examples}
\label{sec:examples}
We begin with two illustrative examples. 
\subsection{Example 1: phase reference frame uncertainty}
\label{sec:example2}
\paragraph{Physical setup.}
Alice and Bob share an optical link along a line of sight; through this link they can perform quantum or classical communication, mediated by individual photons or beams of classical light.  Alice transfers one half of a polarisation-entangled pair of photons to Bob through the optical link, which can be used to teleport the state $\sigma$ of a qubit in her possession. However, they do not share a Cartesian frame defining the $x$- and $y$-polarisation axes in the plane perpendicular to the axis of the link. Due to frame misalignment, Bob's description of the polarisation state of the transmitted photon may differ from Alice's~\cite{Laing2010}.

The reference frame transformation group here is the two-dimensional rotation group $\U(1)$. If $\theta \in [0,2\pi)$ is the angle of a clockwise rotation of the 2D Cartesian frame, $U(1)$ acts as follows on the polarisation state:
\begin{equation}
\label{eq:onlyspeakablecommunication}
\theta \mapsto \rho(\theta)=\begin{pmatrix}
1 & 0 \\ 0 & e^{- i \theta}
\end{pmatrix}
\end{equation}
Here the vector acted on by the matrix is $(v_L, v_R)^T$, where $v_L$ is the left and $v_R$ the right circular polarisation coefficient. The transformation $g(t) \in \U(1)$ which relates Alice and Bob's frames at time $t$ is unknown, and may vary non-negligibly on timescales shorter than the message transmission time between the parties, rendering prior alignment impossible.

\paragraph{Conventional scheme.} Alice creates a polarisation-entangled photon pair 
$$\eta= \frac{1}{\sqrt{2}}(\ket{00}+\ket{11}).$$ She communicates one photon to Bob through the optical link, and measures the the other, together with the state $\sigma$, in the maximally entangled orthonormal basis $\ket {\phi_i} = (\mathbbm{1} \otimes U_i^T) \ket \eta$, where $U_i$ are the Pauli matrices:
\begin{align}
\label{eq:paulis}
U_0 &= \begin{pmatrix} 1 & 0 \\ 0 & 1 \end{pmatrix}
&
U_1 &= \begin{pmatrix} 0 & 1 \\ 1 & 0 \end{pmatrix}
&
U_2 &=  \begin{pmatrix} 0 & -i\\ i & 0 \end{pmatrix}
&
U_3 &= \begin{pmatrix} 1 & 0 \\ 0 & -1 \end{pmatrix}
\end{align}
She communicates the result to Bob through an ordinary classical channel, who applies the correction $U_i$ to his half of the entangled state.
Should both parties' reference frames be aligned, Bob's system will finish in the state $\sigma$; this is because the Pauli matrices form a \emph{unitary error basis} (UEB), a structure we will define later.

However, if Bob's frame is related to Alice's by a nontrivial transformation $g \in \U(1)$, then from the perspective of Alice's frame, Bob will not perform the intended correction $U_i$, but rather the conjugated unitary\footnote{For a proof, see Appendix~\ref{sec:refframetransfrulesproof}.}
\begin{equation}\label{eq:rightconjactionfirstappears}\rho(g)^{\dagger} U_i \rho(g).\end{equation}
\noindent
The transformation $g$ is unknown, so we must average over the whole of $U(1)$ to find the effective channel, yielding the following expression:
\begin{equation}\label{eq:integralfromexample}
\mathcal{T}_i(\sigma) = \int_{U(1)} \d g \;[\rho(g)^{\dagger} U_i \rho(g) U_i^{\dagger}] \, (\sigma)
\end{equation}
Here $\d g$ is the Haar measure on $U(1)$, and we have used the notation $[X](\sigma)$ for the conjugation $X \sigma X^{\dagger}$. Averaging over the four equiprobable measurement results, we  find (Section~\ref{sec:u1calculations}) that the  effective channel for a conventional scheme has the following effect on an input density matrix:
\begin{equation*}
\begin{pmatrix}
a & b \\ c & d
\end{pmatrix} \mapsto
\begin{pmatrix}
a & b/2 \\ c/2 & d
\end{pmatrix}
\end{equation*}

\paragraph{Tight scheme.} Alice measures as before, but now transmits her measurement result using a beam of polarised classical light sent along the optical link, according to the following prescription. If she measures 0 or 3, she transmits a beam of clockwise or anticlockwise circularly polarised light respectively; since the direction of circular polarisation is preserved under reference frame transformations, Bob will receive the measurement result as it was sent. If she measures 1 or 2, she sends the measurement result encoded in the polarisation axis of a beam of linearly polarised light, which is chosen using the regions in Figure~\ref{fig:examplemeasregions}: if she measures 1 or 2, she sends the light linearly polarised along an axis selected uniformly at random from the region $R_1$ or $R_2$ respectively. Bob then observes the polarisation direction of the light he receives respect to his own frame and decodes in the inverse manner, performing the correction as before.  The rationale behind this choice of encoding will be made clear in Section~\ref{sec:theorysection}.

\begin{figure}
\centering
\includegraphics[scale=1]{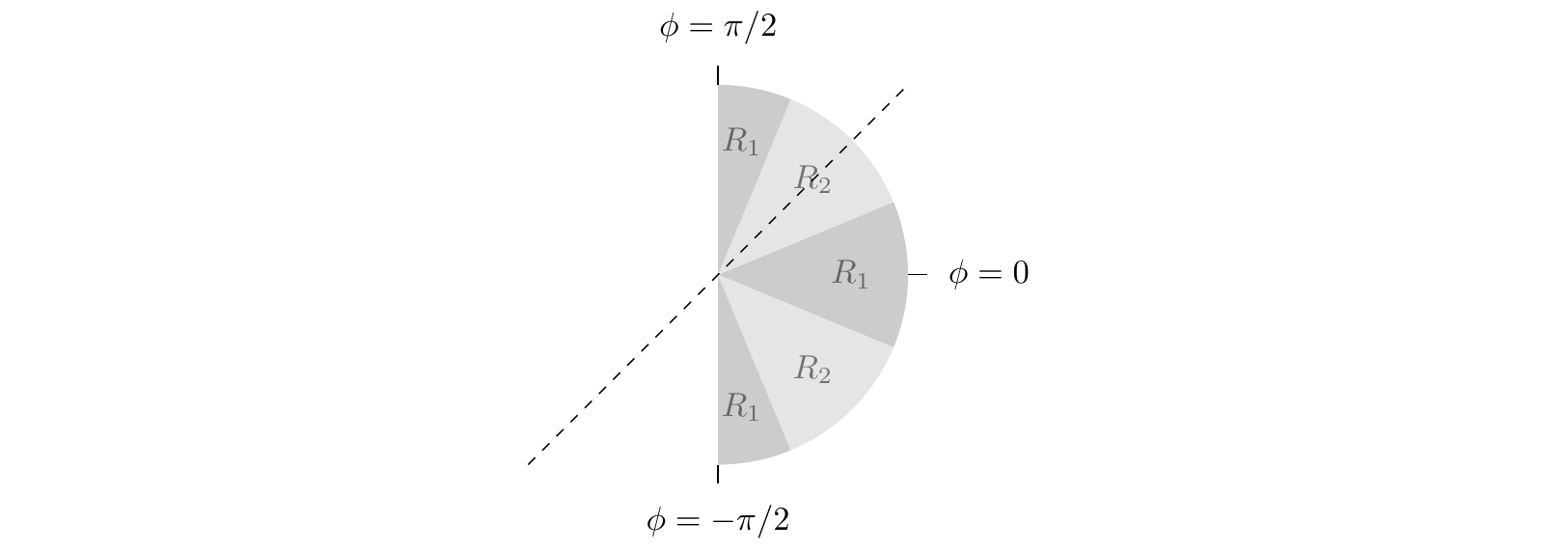}
\caption{The regions $R_1$ and $R_2$. The polarisation axis of a beam of light linearly polarised at angle $\theta=\pi/4$ is shown in the figure.}\label{fig:examplemeasregions}
\end{figure}%

This scheme is \emph{tight}. In particular, we highlight two of the properties listed in Section~\ref{sec:intro}:
\begin{itemize}
\item (NL). To an observer outside Alice's lab, the information she communicates is uniformly random. This follows from the fact that her measurement outcomes are equiprobable, and given the measurement outcome $i$ all polarisation directions in the corresponding region are equiprobable. Therefore, nothing can be deduced from her transmission about her reference frame orientation.
\item (MC). There are four messages Bob can receive: left or right circularly polarised light, or light linear polarised through an axis in the region $E_1$ or $E_2$. All four messages are equiprobable. He therefore obtains precisely two bits of classical information.
\end{itemize}
We will see (Section~\ref{sec:u1calculations}) that the effective channel --- averaging over Alice's equiprobable measurement results --- has the following action on an input density matrix:
\begin{equation}\label{eq:u1dmtwirl}
\begin{pmatrix}
a & b \\ c & d
\end{pmatrix} \mapsto
\begin{pmatrix}
a & b \left(  \frac{2}{\pi^2} + \frac{1}{2} \right)  \\ c \left(  \frac{2}{\pi^2} + \frac{1}{2} \right)  & d
\end{pmatrix}
\end{equation}
The quality of the channel has increased, despite the fact that no reference frame information has been transmitted. In particular, the final state is now asymmetric even when Alice measures 1 or 2.
\ignore{
We remark that this is nothing other than \emph{teleportation of unspeakable information}. The unspeakable classical information transmitted by Alice contains no information about her reference frame configuration, or the state to be teleported; therefore, the correlated unspeakable information in Bob's eventual state can only have been transferred by teleportation. While this may seem at odds with~\cite{Chiribella2012}, the no-go theorem in that work applied only to perfect teleportation with frame--independent measurement result transmission.}

\paragraph{Perfect scheme.} For perfect teleportation, Alice need not transmit full information about the frame in which she measured, as shown by the following scheme. If Alice measures 0 or 3, she transmits a beam of left or right circularly polarised light respectively. If she measures 1 or 2, she transmits linearly polarised light with polar angle $0$ or $\pi/4$ respectively. If Bob receives circularly polarised light, he decodes as before. If he receives linearly polarised light in the region $E_1$ with respect to his own frame, he rotates his frame actively or passively so that the light is polarised along the axis with polar angle $0$ in his frame, and performs the correction $U_1$. If the polarisation direction is in the region $E_2$, he rotates his frame actively or passively so that the light is polarised along the axis with polar angle $\pi/4$ in his frame, and performs the correction $U_2$. We will see (Proposition~\ref{prop:perfecttel}) that this procedure results in perfect teleportation. However, the reference frame information communicated by this protocol is only sufficient to reduce reference frame uncertainty to a finite subgroup $\mathbb{Z}_4$.
\ignore{
This scheme possesses the (DR) property, since the parties' reference frame orientation need only remain approximately constant between Bob's receipt of the classical information and his performance of the corresponding correction; it clearly also uses an entangled resource of minimal dimension, and so satisfies (ME). However, the procedure violates (MC) and (NL), since Bob gains an infinite amount of information about Alice's reference frame alignment upon measurement.
}

\subsection{Example 2: spatial reference frame uncertainty}\label{sec:example}

\paragraph{Physical setup.}
Alice and Bob are spatially separated; their qubits are spin-$\frac{1}{2}$ particles. Alice plans to teleport a state $\sigma$ to Bob.  They each possess half of the following maximally entangled pair\footnote{Note that the entangled state is invariant under changes in reference frame, so both parties' frames may shift arbitrarily following its creation without affecting the quality of the entangled resource.}:
$$
\ket{\eta} = \frac{1}{\sqrt{2}}( \ket{01} - \ket{10})
$$
However, the Cartesian frame according to which Alice's $x$-, $y$- and $z$-spin axes are defined is related to Bob's by some unknown three-dimensional rotation. The reference frame transformation group is $\SU(2)$, which acts on a qubit Hilbert space $H$ by its standard matrix representation $\rho: \SU(2) \to B(H)$. Again, the transformation $g(t) \in \SU(2)$ which relates Alice's and Bob's frames at time $t$ is unknown, and may vary on timescales shorter than the message transmission time between the parties.

\paragraph{Conventional scheme.} Alice and Bob use the entangled state $\ket{\eta}$ to attempt a standard teleportation protocol~\cite{Bennett1993}, again based on the Pauli matrices~\eqref{eq:paulis}. Alice measures the state $\sigma$ together with her entangled qubit in the maximally entangled orthonormal basis $\ket {\phi_i} = (\mathbbm 1 \otimes -i(  U_i U_2)^T) \ket \eta$,\footnote{The $-i$ and the $U_2$ here correspond to the choice of maximally entangled state; see the discussion following Theorem~\ref{Wernertighttheorem}.} and communicates the measurement result to Bob through an ordinary classical channel; Bob then applies the correction $U_i$. We must average over all misalignments in $\SU(2)$ to find the effective channel. For measurement result $i$ we obtain the following expression:
\begin{equation}\label{eq:integralfromexample}
\mathcal{T}_i(\sigma) = \int_{\SO(3)} \d g \;[\rho(g)^{\dagger} U_i \rho(g) U_i^{\dagger}] \, (\sigma)
\end{equation}
Here $\d g$ is the Haar measure on $\SO(3)$. Averaging over the four equiprobable measurement results, we  find (Section~\ref{sec:su2calculations}) that the effective channel purity is approximately $0.21$.

\ignore{
However, if Bob's frame is related to Alice's by a nontrivial transformation $g \in \SU(2)$, then from the perspective of Alice's frame, Bob will not perform the intended correction $U_i$, but rather the conjugated unitary\footnote{For a proof, see Appendix~\ref{sec:refframetransfrulesproof}.}
\begin{equation}\label{eq:rightconjactionfirstappears}\rho(g)^{\dagger} U_i \rho(g).\end{equation}
\noindent
Since the conjugation action of $\SU(2)$ has kernel $\{ \pm I\}$, we only consider the quotient $\SO(3)$ in the following analysis. The transformation $g$ is unknown, so we must average over the whole of $\SO(3)$ to find the effective channel, which for measurement result $i$ yields the following expression:
\begin{equation}\label{eq:integralfromexample}
\mathcal{T}_i(\sigma) = \int_{\SO(3)} \d g \;[\rho(g)^{\dagger} U_i \rho(g) U_i^{\dagger}] \, (\sigma)
\end{equation}
Here $\d g$ is the Haar measure on $\SO(3)$, and we have used the notation $[X](\sigma)$ for the conjugation $X \sigma X^{\dagger}$. Averaging over the four equiprobable measurement results, we  find (Section~\ref{sec:su2calculations}) that the effective channel purity is approximately $0.21$.}

\paragraph{Tight scheme.}
Alice considers a cube centered at the origin of her frame, oriented so that the $x$-, $y$- and $z$-axes form normal vectors to its faces; we call the faces intersected by the $x$-, $y$- and $z$-axes the $1$-, $2$- and $3$-faces respectively. She measures in the basis $\{\ket{\phi_i}\}$, and transmits her measurement result using the encoding scheme given in Table~\ref{tbl:typeaencodingscheme}, and illustrated in Figure~\ref{fig:cubes}, which we summarize as follows.
\begin{table}
\centering
\begin{tabular}{|l|l|l|}

\hline \bf Measurement result & \bf Classical transmission \\
 \hline
0 & Featureless sphere \\
1 & Rod oriented along any axis intersecting the $1$-faces \\
2 & Rod oriented along any axis intersecting the $2$-faces \\
3 & Rod oriented along any axis intersecting the $3$-faces \\\hline
\end{tabular}
\caption{Tight encoding scheme for the rod channel. Alice chooses the precise orientation of the rod uniformly at random from the set of all orientations satisfying the intersection condition.}
\label{tbl:typeaencodingscheme}
\end{table}
If Alice receives measurement result 0, she sends a spherically symmetric object (in other words, a sphere) to Bob. Otherwise, if she receives measurement result $n \in \{1,2,3\}$, she prepares a rigid rod in an arbitary orientation in space, centred at the origin of her frame, such that it intersects the $n$-faces of the cube. She then sends this object to Bob by parallel transport.

When Bob receives the object from Alice, he performs the reverse of Alice's encoding scheme. If he receives the spherically symmetric object he performs correction $U_0$. If he receives a rod, he moves it by parallel transport to his origin, and observes which faces of the cube it intersects. Bob's cube will of course in general be oriented differently to Alice's, and so he may observe a different intersection than that encoded by Alice. Having observed an intersection with the $n$-faces, he then performs correction $U_n$.

\begin{figure}
\centering
\includegraphics[scale=1]{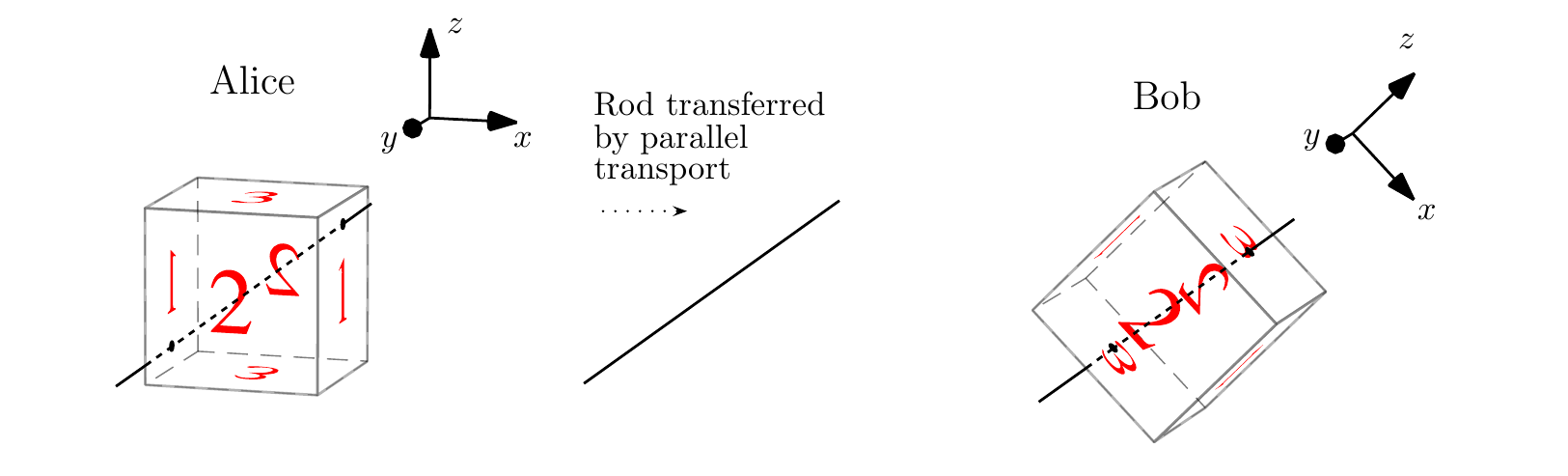}
\caption{Tight encoding scheme for  the rod channel. Alice measures 1, chooses at random an orientation of the rod which intersects the $1$-faces of the cube in her frame, and communicates the rod to Bob by parallel transport along a straight path. In Bob's frame, related to Alice's by a $\pi/4$-rotation around the $y$-axis, the rod intersects the $3$-faces; he therefore performs the correction $U_3$.}
\label{fig:cubes}
\end{figure}

In Section~\ref{sec:su2calculations} we numerically calculate the purity of the effective channel as $ 0.44 \pm 0.03$, approximately double the value for a conventional scheme.

This scheme is \emph{tight}, possessing in particular the  (NL) and (MC) properties, for exactly the same reasons as the previous example.
\ignore{
\begin{itemize}
\item (ME), (ML). Immediate.
\item (DR). The effective channel~\eqref{eq:type1integralexample} is unaffected by  changes in reference frame orientation during execution, as long as Bob's reference frame does not change between observing the rod in his lab and performing the corresponding correction.
\item (NL). To an observer outside Alice's lab, the information she communicates is uniformly random. This follows from the fact that her measurement outcomes are equiprobable, and given the measurement outcome $i$ all directions through the corresponding face pair of the cube are equiprobable. Therefore, nothing can be deduced from her transmission about her reference frame orientation.
\item (MC). There are four messages Bob can receive: a spherically symmetric object, or a rod oriented through the $i$-faces for some $i \in \{1,2,3\}$. All four messages are equiprobable. He therefore obtains precisely two bits of classical information.
\end{itemize}}

\paragraph{Perfect scheme.} Again, transmission of a full reference frame is unnecessary for perfect teleportation. We call the following family of unitary matrices the \emph{tetrahedral} qubit unitary error basis~\cite{Verdon2018}:
\begin{align}
\nonumber
V_0 &= \scriptsize \begin{pmatrix} 1 & 0 \\ 0 & e^{2 \pi i /3} \end{pmatrix}
&
V_2 &= \frac{1}{\sqrt{3}} \scriptsize\begin{pmatrix} 1 & \sqrt{2}e^{2 \pi i /3} \\ \sqrt 2 & e^{5 \pi i /3} \end{pmatrix}
\\[-9pt]
\label{eq:exampleueb}
\\[-5pt]
\nonumber
V_1 &= \frac 1 {\sqrt{3}} \scriptsize \begin{pmatrix} 1 & \sqrt{2}e^{4 \pi i/3}\\ \sqrt{2}e^{4 \pi i/3} & e^{5 \pi i/3} \end{pmatrix}
&
V_3 &= \frac {1} {\sqrt{3}} \scriptsize \begin{pmatrix} 1 & \sqrt{2} \\ \sqrt{2}e^{2\pi i/3} & e^{5 \pi i/3} \end{pmatrix}
\end{align}
Let $\Tet \subset \SO(3)$ be the subgroup preserving a regular tetrahedron centred at the origin with vertices:
\begin{align*}
v_0 = \hat z && v_1 = \frac 1 {3} (\sqrt{8}\hat x -\hat z) && v_2 = \frac 1 {3} (- \sqrt{2} \hat x + 2\sqrt{3} \hat y - \hat z) && v_3 = \frac 1 {3} (- \sqrt{2} \hat x - 2 \sqrt{3} \hat y - \hat z)\end{align*}
\noindent
We identify the elements of $\Tet \cong \SO(3)$ with the permutation they induce on these vertices.

Alice again measures in the basis $\{\ket {\phi_i}\}$, where $\ket{\phi_i} = (\mathbbm 1 \otimes -i(  V_i U_2)^T) \ket \eta$. To perform the classical communication, Alice uses a completely asymmetric classical object whose orientation exactly determines a frame of reference. In order to transmit the measurement result $i$, she aligns the asymmetric object so that the frame determined by its orientation matches her own Cartesian frame. She then rotates the object by an element $r^A \in \Tet$, according to the prescription in Table~\ref{tbl:typecexample}, and sends it to Bob. 
\begin{table}
\centering
\begin{tabular}{l | l | l }
Measurement result & Alice's rotation $r^A$ &  Bob's observation $r^B$ \\
\hline \hline 
0 & $()$ & $()$ or $(234)$ or $(243)$ \\
1 & $(132)$ & $(142)$ or $(132)$ or $(12)(34)$ \\
2 & $(123)$ & $(13)(24)$ or $(123)$ or $(143)$\\
3 & $(134)$ & $(134)$ or $(124)$ or $(14)(23)$
\end{tabular}
\caption{Type C encoding scheme for the matched channel.}
\label{tbl:typecexample}
\end{table}
Bob observes the orientation of the object according to his own Cartesian frame, and realigns his frame (actively or passively) by the smallest possible angle so that the rotation $r^B$ taking his frame onto that determined by the orientation of the asymmetric object is in $\Tet$. He then uses Table~\ref{tbl:typecexample} to decide which measurement result $j$ to correct for, and performs --- in his own frame --- the correction $ V_j$.

While this procedure only reduces reference frame uncertainty to the \emph{binary tetrahedral} subgroup of $\SU(2)$, it will be shown in Proposition~\ref{prop:perfecttel} that it results in perfect teleportation. As before, it possesses the (DR) and (ME) properties, but violates (MC) and  (NL).

\section{Theory}
\label{sec:theorysection}
We now explain the theory behind the examples in Section~\ref{sec:examples}.

\subsection{Equivariant unitary error bases}
We first recall the notion of a unitary error basis. 
\begin{definition}\label{def:ueb}
A \emph{unitary error basis} (UEB) for a $d$-dimensional Hilbert space $V$ is a basis of $d^2$ unitary matrices $\{U_i\}_{i \in I}$ in $B(V)$ (where $I = \{1, \dots, d^2\}$ is the index set) which is orthonormal under the Hilbert-Schmidt inner product: \begin{equation}
\braket {U_i} {U_j} := \frac{1}{d} \Tr(U_i^{\dagger}U_j) = \delta_{ij}
\end{equation}
\end{definition}
\begin{theorem}[{\cite[Theorem 1]{Werner2001}}]
\label{Wernertighttheorem}
A teleportation protocol satisfying the (ME) property corresponds to a choice of unitary error basis for $V$, along with any other unitary matrix~$X$.
\end{theorem}
\noindent
Under this correspondence, the shared entangled state $\eta$ is the maximally entangled state $\sum_i \ket{i} \otimes X\ket{i}$ for a chosen orthonormal basis $\{\ket{0},\ket{1},\dots\}$ and some unitary $X$. (Any bipartite maximally--entangled pure state is of this form.) Alice measures in the maximally--entangled orthonormal basis $\{\ket{\phi_i}\}_{i\in I}$, where \begin{equation}\label{eq:phimeas} \ket{\phi_x} = \sum_i \ket{i} \otimes (U_x X)^{T} \ket{i}. \end{equation}
Bob's correction for measurement outcome $x$ is $U_x$. 

We now consider the effect of reference frame misalignment on such a procedure. Let $G$ be a compact Lie group of reference frame transformations, with unitary representation $\rho: G \to B(V)$ on Bob's system; here and throughout we assume uniform reference frame uncertainty, where the probability measure over $G$ is the Haar measure $\text dg$. We assume that the maximally entangled state $\ket{\eta} \in V \otimes V$ is invariant up to a phase under changes in frame, so that the entanglement is not itself degraded by reference frame uncertainty.\footnote{The existence of such states is treated in an appendix of our earlier work~\cite{Verdon2018}.} 
We work in Alice's frame. In this frame, Alice performs the measurement correctly and sends the result $i$, but Bob performs the correction $\rho(g)^{\dagger} U_i \rho(g)$.\footnote{For a proof, see Appendix~\ref{sec:refframetransfrulesproof}.} Since $g \in G$ is unknown, the effective channel when Alice measures $i$ is
\begin{equation}\label{generalnaivechanneleqn}
\sigma'_i = \int_G \text dg  \; [\rho(g)^{\dagger} U_i \rho(g) U_i^{\dagger}] (\sigma) .
\end{equation}
For finite $G$, we can use an \emph{equivariant} UEB together with a classical channel carrying a $G$-action to perform perfect reference frame--independent teleportation.\footnote{The existence of equivariant UEBs is treated in~\cite{Verdon2018}, with a complete classification for qubit systems.}
\begin{definition}
Let a finite group $H$ act on a Hilbert space $V$ of dimension $d$ by the representation $\rho: H \to B(V)$. We say that a unitary error basis $\{U_i\}_{i \in I}$ for $V$ is \textit{$H$-equivariant} when the right conjugation action of $H$ permutes the elements of $\{U_i\}_{i \in I}$ up to a phase. Explicitly,
$$
\rho(h)^{\dagger} U_i \rho(h) = \alpha(i,h) U_{\sigma(i,h)} \;\;\; \forall \, h \in H,i \in I
$$
where $\sigma: I \times H \to I$ is a right action of $H$ on the index set $I = \{0, \dots, d^2-1\}$, and $\alpha: I \times H \to \U(1)$ is some phase.
\end{definition}
\begin{proposition}[{\cite[Theorem 2.7]{Verdon2018}}]\label{prop:finitegroupseqtel}
Let $H$ be a finite group of reference frame transformations with an equivariant unitary error basis $\{U_i\}_{i \in I}$ and corresponding right action $\sigma: I \times H \to I$. Let Alice communicate the measurement results using a channel whose set of messages $I$ carries the inverse left action $\sigma^{-1}: H \times I \to I$. Then the teleportation protocol with data $\{U_i\}_{i \in I}$ will function perfectly for all $h_{AB} \in H$.
\end{proposition}
\begin{proof}
In Alice's frame, for measurement result $i$ and any misalignment $h \in H$, Bob will perform the correction $
\rho(h)^{\dagger} U_{\sigma^{-1}(h,i)} \rho(h) \sim U_{\sigma(\sigma^{-1}(h,i),h)} = U_i$.
\end{proof}
\noindent Here we consider actions of general (i.e. possibly infinite) compact Lie groups $G$, for which equivariant UEBs generally do not exist. Our approach here is to identify a finite subgroup $H\subset G$ such that there exists an equivariant UEB for $H$ under the restricted representation. We then choose an encoding of the measurement result in the classical channel which carries the inverse action in the sense of Proposition~\ref{prop:finitegroupseqtel}, allowing us to  `quotient' the space of possible misalignments $G$ by the subgroup $H$.

\begin{remark}
If the representation of $G$ on the system to be teleported is not faithful, we can consider the natural faithful representation of the \emph{reduced reference frame transformation group} $\tilde{G}:=G/\Ker(\rho)$. In Section~\ref{sec:example2}, for instance, the reduced transformation group was $U(1)/{\mathbb{Z}_2} \cong U(1)$, because the representation~\eqref{eq:onlyspeakablecommunication} obeys $\rho(2 \theta) = \rho(\theta)$. For the faithful action, we can use the results about existence of equivariant UEBs from~\cite{Verdon2018}. We cannot simply assume that $G$ acts faithfully, though, since when constructing a compatible classical channel  it will be necessary to consider the physical rather than the reduced transformation group.
\end{remark}

\begin{example}
\begin{itemize}
\item The UEB in both the tight and perfect schemes for $\U(1)$ (Section~\ref{sec:example2}) is the set of Pauli matrices, which is equivariant for the subgroup $\mathbb{Z}_4<U(1)$ of the reduced transformation group. A generator of $\mathbb{Z}_4$ acts as the swap $(12)$ on the index set of the UEB under conjugation.
\item In the tight scheme for $\SU(2)$ (Section~\ref{sec:example}) the Pauli UEB is equivariant for the binary octahedral subgroup $\BOct \subset  \SU(2)$ preserving the cube.  
\item In the perfect scheme for $\SU(2)$ (Section \ref{sec:example}) the tetrahedral UEB is equivariant for the binary tetrahedral subgroup $\BTet \subset \SU(2)$ preserving the tetrahedron. \end{itemize}
\end{example}

\subsection{Compatible encoding of classical information}\label{sec:encodingschemes}

We now consider the other component of the scheme, a classical channel carrying an action of the reference frame transformation group. The spaces of readings of all the classical channels we consider in this work carry a smooth manifold structure with normalised measure $\text dx$, and all actions are smooth and measure-preserving.
 
\begin{definition}\label{def:unspeakablecomm}
We say that a classical channel \emph{communicates unspeakable information}~\cite{Peres2002}, or is an \emph{unspeakable channel}, if its space of readings $C$ carries a nontrivial action of the reference frame transformation group $G$.

We call a channel whose space of readings carries a trivial $G$-action a \emph{speakable channel}.
\end{definition}
\noindent 
Throughout this paper we make the simplifying assumption that there is no channel noise, apart from that arising from frame misalignment. A classical channel is therefore fully described by its space of readings and the $G$-action on that space; for this reason we conflate the channel with its space of readings, using the same letter $C$ for both.
Since we have chosen the convention that the effect of a change of reference frame on the states of a quantum system corresponds to a left action of the transformation group (see Appendix~\ref{sec:refframetransfrulesproof}), the action of $G$ on the classical channel will be a left action.
\begin{example}
\begin{itemize}
\item For the tight and perfect schemes in Section~\ref{sec:example2}, the space of readings was the linear polarisation direction of the light beam.  As a smooth manifold, this is the real projective line $\mathbb{RP}^1$; it carries a non-faithful smooth action of $U(1)$ with kernel $\mathbb{Z}_2$ (since a $\pi$ rotation does not change the polarisation direction).
\item For the tight scheme in Section~\ref{sec:example}, the space of readings was the space of possible orientations of a rod. As a manifold, this is the real projective plane $\mathbb{RP}^2$, carrying the obvious smooth action of $SU(2)$. 
\item For the perfect scheme in Section~\ref{sec:example}, the space of readings was the space of possible orientations of a completely asymmetric object. As a manifold, this is the Stiefel (frame) manifold $V_{2}(\mathbb{R}^3) \cong SO(3)$, carrying the obvious smooth action of $SU(2)$.
\end{itemize}
\end{example}
We now specify a framework for encoding of measurement values in such a channel.
\begin{definition}[Encoding scheme]\label{def:encodingscheme}
Let $C$ be an unspeakable channel and $I$ be a finite set of values to be sent through it. An \emph{encoding scheme} for $I$ is:
\begin{itemize}
\item A set of open subsets $\{E_i \subset C \; | \; i \in I\}$, the \emph{encoding subsets}, where $E_i$ are disjoint open sets.
\item A set of open subsets $\{D_i \subset C \; | \; i \in I\}$, the \emph{decoding subsets}, where $D_i$ are disjoint open sets which cover $C$ up to a set of measure zero.
\end{itemize}
The encoding subset $E_i$ is the set of all possible readings Alice can send in order to transmit the value $i \in I$. The decoding subset $D_i$ is the set of all possible readings upon receipt of which Bob will record the value $i \in I$.
\end{definition}
\noindent Recalling Proposition~\ref{prop:finitegroupseqtel}, the success of our protocol depends on encoding schemes which are compatible with the right action of $H$ on the index set of the UEB.
\begin{definition}[Compatible channel]\label{def:compatiblechannel}
Let $C$ be an unspeakable channel for a finite group $H$. Let $\sigma: I \times H \to I$ be a right action of $H$ on an index set $I$. We say that an encoding scheme for $I$ is \emph{compatible with $\sigma$} if:
\begin{itemize}
\item The decoding subsets $\{D_i\}_{i \in I}$ and the encoding subsets $\{E_i\}_{i \in I}$ are each permuted under the action of $H$ on $C$, inducing left actions $\tau_D, \tau_E: H \times I \to I$.
\item The left actions $\tau_D, \tau_E: H \times I \to I$ are equal and inverse to the action $\sigma: I \times H \to I$  of $H$ on $I$. That is, for all $i \in I$, $$\tau_D(i, -) = \tau_E(i, -) = \sigma^{-1}(i,-).$$
\end{itemize}
\end{definition}
\noindent In words: given a right action of a finite reference frame transformation group on the UEB index set, a compatible encoding scheme transmits the indices through the classical channel with the inverse left action.
\begin{example}
\begin{itemize}
\item In Section~\ref{sec:example2}, the encoding and decoding subsets for the tight scheme are the same, namely the regions $R_1$ and $R_2$ (Figure~\ref{fig:examplemeasregions}). In the physical (unfaithful) representation, the Pauli UEB is equivariant for the subgroup $\mathbb{Z}_8< U(1)$, where a generator of $\mathbb{Z}_8$ acts as the swap (12). Compatibly, the regions $R_1$ and $R_2$ are swapped under the action of a generator of $\mathbb{Z}_8$. For the perfect scheme, the encoding subsets are singletons, namely the polar angles $0$ and $\pi/4$; the decoding subsets are the regions $R_1$ and $R_2$.
\item In the tight scheme of Section~\ref{sec:example}, the encoding and decoding subsets are the same: $D_i = E_i$ is the subset of orientations of the rod through the $i$-faces of the cube.  The indices of the Pauli UEB are permuted inversely to the labels on the cube's faces under the conjugation action of $\BOct$. 
\item In the perfect scheme of Section~\ref{sec:example}, the encoding subsets $E_i$ are singletons, namely the orientations given by rotating  the object according to Table~\ref{tbl:typecexample}. The decoding subsets are Voronoi cells around these orientations~\cite{Yan2013}. The indices of the tetrahedral UEB are permuted inversely to the encoding and decoding subsets under the conjugation action of $\BTet$.
\end{itemize}
\end{example}

\subsubsection{A construction of compatible encoding schemes}
\label{sec:referenceframeencodings}
We now provide a general construction of a compatible encoding scheme for any transitive action $\sigma: I \times H \to I$ of a finite subgroup of $G$. Since all actions split into transitive actions on the orbits, this loses no generality, since we can  communicate the orbit index using speakable communication. For the construction, we need an unspeakable classical channel of the following type.  Recall that an action is \emph{free} if all stabilisers are trivial, and \emph{transitive} if it possesses only one orbit.
\begin{definition}
Let $G$ be the reference frame transformation group, with representation $\rho$ on the system to be teleported. Let $C$ be an unspeakable classical channel, carrying the action $\alpha: G \times C \to C$. We say $C$ is \emph{matched to $\rho$} if $\Ker(\rho) \subseteq \Ker(\alpha)$, and the reduced action $\tilde{G} \times C \to C$, where  $\tilde{G} = G/\Ker(\rho)$ is the reduced transformation group, is free and transitive. 
\end{definition}
\begin{example}
\label{ex:rfchannels}
\begin{itemize}
\item In Section~\ref{sec:example2} the kernel of the representation $\rho$ is $\mathbb{Z}_2$, generated by the rotation through an angle $\pi$. Likewise, the kernel of the action of $\U(1)$ on polarisation directions is $\U(2)$. The reduced group $G/\Ker(\rho)$ corresponds to the rotations $\theta \in (-\pi/2,\pi/2]$, which clearly act freely and transitively on the polarisation directions. 
\item The channel for the perfect scheme in Section~\ref{sec:example}, where a completely asymmetric classical object was transmitted, is a matched channel for the representation of $\SO(3)$. Here the kernel of $\rho$ is trivial, and the action of $\SO(3)$ on the set of orientations is clearly free and transitive.
\ignore{
\item Here is an example relevant for teleportation of a quantum clock~\cite{Chiribella2012}. Consider the group $U(1)$ of time translations of a quantum system; a frame corresponds to a choice of a `zero of time' $t_0$. Any fixed event perceived by both parties defines a matched channel, whose space of readings corresponds to the time of the event with respect to the period of evolution of the quantum system.}
\end{itemize}
\end{example}
\noindent

The readings of a matched channel $C$ can be identified with elements of the reduced transformation group $\tilde{G}$, by choosing an `identity' reading $[e] \in C$ based on their own reference frame configuration. All other readings in $C$ are then identified uniquely by $[g]:= g \cdot [e]$, for any $g \in \tilde{G}$.
\begin{example}
\begin{itemize}
\item For the channel of Section~\ref{sec:example2}, the channel reads $[e]$ when the polarisation axis is the $x$-axis of the observer. 
\item For the perfect scheme of Section~\ref{sec:example}, the the channel reads $[e]$ when the frame defined by the asymmetric object is aligned with the Cartesian frame of the observer.  
\ignore{
\item For the time-of-event channel in Example~\ref{ex:rfchannels}, the channel reads $[e]$ when the event occurs at the observer's zero of time $t_0$ with respect to the period.}
\end{itemize}
\end{example}
In general, Alice and Bob will have different labellings of the channel, given that their reference frames are oriented differently. We write $[g]_A, [g]_B$ for the reading associated to $g \in G$ by Alice and Bob respectively. 
\begin{proposition}\label{prop:readingtransfunderrfchange}
If Bob's frame is related to Alice's by a transformation $g_{AB} \in G$, then their labellings are related as follows:
\begin{equation}\label{eq:readingtransfsunderrfchange}
[g]_A = [ g g_{AB}^{-1}]_B
\end{equation}
\end{proposition}

We now construct the compatible encoding scheme. We recall the following characterisation of transitive actions.
\begin{lemma}\label{lem:cosetspaces}
Let $H$ be a finite group. Any transitive right $H$-set is isomorphic to a right coset space $L \backslash H$ for a subgroup $L \subset  H$ under the right action $(Lh_2)\cdot h_1 = Lh_2h_1$. 
\end{lemma}
\noindent 
Our construction divides the matched channel $C$ up into regions $\{R_h \subset C \;|\; h \in H\}$, which are permuted by reference frame transformations in $H$ according to the inverse left action $h_2 \cdot R_{h_1} = R_{h_1 h_2^{-1}}$. We then identify these regions to obtain the desired transitive action. To define the $R_h$, we choose a fundamental domain for the finite subgroup $H \subset \tilde{G}$.
\begin{definition}
A \emph{fundamental domain} for a finite subgroup $H \subset G$ is an open subset $F \subset G$ containing the identity such that the $H$-translates $Fh$ have empty intersection and cover $G$ up to a set of measure zero. \footnote{It is good to pick $F$ so that all the readings in it are as close to the identity as possible under some metric. To make this precise one can use \emph{Voronoi cells}~\cite{Yan2013}.}
\end{definition}
\begin{example}
In the example of Section~\ref{sec:example2}, the rotations through an angle $\theta \in (-\pi/8,\pi/8)$ are a fundamental domain for $\mathbb{Z}_4 \subset \tilde{G}$.
\end{example}
\begin{definition}\label{def:blockdivision}
Fix a subgroup $H\subset G$, and a fundamental domain $F$ for $H$ in $G$. Then the regions $\{R_h \;|\; h \in H\}$ are defined as
$$R_h := \{[fh]\, |\, f \in F\}.$$
\end{definition}
\begin{lemma}\label{lem:blocktransfunderrfchange}
Let Bob's reference frame configuration be related to Alice's by a transformation $h_{AB} \in H$. Then 
\begin{equation*}
(R_h)_A = (R_{h h_{AB}^{-1}})_B.
\end{equation*}
\end{lemma}
\begin{proof}
Immediate from~\eqref{eq:readingtransfsunderrfchange}.
\end{proof}
We can now construct a compatible encoding scheme for the transitive action $L \backslash H$ by grouping regions $R_h$ into cosets. Let $c_i \in H$ be right coset representatives for $L$ in $H$.
\begin{definition}\label{def:rfencodingschemes}
The \emph{tight matched scheme} for $\sigma$ is defined as:
\begin{align*}
D_i = \bigsqcup_{l\in L} R_{lc_i} && E_i = D_i
\end{align*}
The \emph{perfect matched scheme} is defined as:
\begin{align*}
D_i &= \bigsqcup_{l\in L} R_{lc_i} &&
E_i = \{ \bigsqcup_{l\in L} [lc_i] \}
\end{align*}
\end{definition}
\noindent
The reason for the nomenclature will become apparent in the next section.
\ignore{
\noindent It is clear that these subsets are disjoint and open and that the $\{[D^k_i]\}_{i \in I_k}$ cover the channel up to a set of measure zero; this is therefore an encoding scheme. We now show compatibility.
\begin{proposition}\label{prop:(f,h,k)action}
The encoding schemes of Definition~\ref{def:refframeencodingschemeprescription} are compatible with $(I_k, \sigma_k)$. That is, the subsets $(D^k_i,E^k_i)_{i \in I_k}$ carry the following action of $H$ under reference frame changes $h_{AB} \in H$:
\begin{align*}
[D^k_i]_A
= [D^k_{\sigma^{-1}(h_{AB},i)}]_B && [E^k_i]_A
= [E^k_{\sigma^{-1}(h_{AB},i)}]_B
\end{align*}
\end{proposition}
\begin{proof}
We have $[D_i]_A =  \sqcup_{k\in K} [F(kc_i)]_A =  \sqcup_{k\in K} [F(kc_ih_{AB}^{-1})]_B =  \sqcup_{k\in K}  [F(k\alpha(\sigma^{-1}(h_{AB}, i))]_B  = \sqcup_{k\in K}  [F(kc_{\sigma^{-1}(h_{AB}, i)}]_B = [D_{\sigma^{-1}(h_{AB},i)}]_B$. The result for $[E_i^k]$ in the perfect encoding follows similarly.
\end{proof}
}
\subsection{Teleportation schemes}\label{sec:procedures}

We now specify and prove correctness for our teleportation schemes. Throughout this section, let $H\subset G$ be a finite subgroup, let $\{U_i\}_{i \in I}$ be an equivariant UEB for $H$, let $\sigma: I \times H \to I$ be the corresponding right action of $H$ on the index set of the UEB, let $I_k \subset I$ be the orbits in $I$ under $\sigma$, where $k$ is some index for the orbits, and let $\sigma_k: I_k \times H \to I_k$ be the corresponding (transitive) restricted actions. 

\ignore{ 
Our approach depends on the existence of a compatible encoding scheme for the action $\sigma: I \times H \to I$ induced by the right conjugation action of the finite subgroup $H \subset G$ on an equivariant UEB with index set $I$. In Section~\ref{sec:referenceframeencodings}, we will show, using an unspeakable classical channel corresponding to a shared reference frame system, a compatible encoding scheme for any \emph{transitive} right action $\sigma$ of $H$ may be constructed. 

In order that our procedures can be applied to nontransitive actions, we use the orbit splitting of $I$ under the $H$-action. Alice will first communicate, through a speakable channel, the orbit $O \subset I$ of her measurement result; she will then then communicate the measurement index $i \in O$ using an unspeakable classical channel with the set of messages $O$, compatible with the restricted action $\sigma|_{O}:O \times H \to O$, which is transitive and therefore amenable to our construction in Section~\ref{sec:referenceframeencodings}. We will see that this does not affect any of the desirable properties of the teleportation schemes. Of course,  if one can find an equivariant UEB with a single orbit under the $H$-action~\cite{Verdon2018}, such as the tetrahedral UEB for $\BTet < \SU(2)$ in Section~\ref{sec:example}, or combine different orbits in a single physical channel, as in Section~\ref{sec:example2}, this prior speakable communication of the orbit label is unnecessary.
}

\subsubsection{Tight scheme}

\begin{procedure}[Tight teleportation scheme]\label{proc:meastransferproc}
Let $C$ be an unspeakable channel for $G$ (and therefore also for $H$), and let $(D^k_i,E^k_i)_{i \in I}$ be encoding schemes for $I_k$ on $C$ compatible with $\sigma_k: I_k \times H \to I_k$ and such that, for each $k$, the decoding regions are the same as the encoding regions, that is, $D_i^k = E_i^k$ for all $i,k$.

Alice measures in the basis $\{\ket{\phi_i}\}_{i \in I}$ \eqref{eq:phimeas} as in a standard teleportation protocol, and obtains the result $i \in I_k$. The result is transmitted as follows.
\begin{enumerate}
\item Alice transmits the orbit label $k$ through a speakable channel.
\item Alice sends a reading $x$ chosen uniformly at random from the region $E_i^k$.
\item Bob receives $g \cdot x \in D_j^k$ and performs the correction $U_j$.
\end{enumerate}
Here $g$ is the reference frame transformation taking Alice's frame at the time of measurement onto Bob's frame at the time of receipt.
\end{procedure}
\noindent
We now derive an explicit expression for the effective channel obtained using Procedure~\ref{proc:meastransferproc}. Recall that, for operators $M, \sigma \in B(H)$, we write $[M](\sigma)$ for $M \sigma M^{\dagger}$.
\begin{theorem}[Effective channel for Procedure~\ref{proc:meastransferproc}]\label{thm:maintheorem}
Suppose that Alice measures some result $i \in I_k$, where $D_k^i = E_k^i$ for all $i \in I_k$. Then the channel induced by Procedure~\ref{proc:meastransferproc} is:
\begin{equation}\label{eq:maintheorem}
\mathcal{T}_k(\sigma)= \frac{|I_k|}{\mu_C(E_0^k)} [\rho(c_i)] \circ \int\limits_G \left( \emph{d}g\, p(g)\, [\rho(g)^{\dagger} U_0  \rho(g)U_0^{\dagger}]  \circ [\rho(c_i)^{\dagger}] \;(\sigma) \right)
\end{equation}
Here $0 \in I_k$ is some fixed element of the orbit;  the normalising factor $\mu_C(E_0^k)$ is the measure of $E_0^k$ in $C$; $p(g) = \int_{E_0^k \subset C} \emph dx \, \mathbbm{1}_{D_0^k}(g \cdot x)$, where $\mathbbm{1}_{D_0^k}$ is a continuous approximation to the indicator function for $D_0^k \subset C$; and $\{c_i\}_{i \in I_k}$, $c_i \in H$ are such that $c_i \cdot E^k_0 =E^k_i$.
\end{theorem}
\begin{proof}
The proof is somewhat technical, so has been placed in Appendix~\ref{app:effchanexpproof}.
\end{proof}
\begin{proposition}\label{prop:meastransferprocproperties} Procedure~\ref{proc:meastransferproc} satisfies (MC), (NL), (ME) and (DR).
\end{proposition}
\begin{proof}
(NL): Alice has an equal probability of measuring any $i \in I_k$, and chooses a reading with uniform probability from the subsets $\{E_i^k = D^i_k\}_{i \in I}$, which have equal measure and cover the space of readings up to a set of measure zero. The message therefore communicates no information about Alice's frame configuration, since without prior knowledge of the reading Alice sent, nothing can be learned from the reading that is received.

(MC): The only useful information Bob learns from the message he receives is which of his decoding subsets $\{D_i^k\}_{i \in I_k}$ the reading he receives lies in; there are $\sum_k |I_k| = |I| = d^2$ possible messages, which are equiprobable. In total, therefore, he receives two dits of unspeakable classical information.

(ME): Obvious.

(DR) In Alice's frame, reference frame misalignment affects Bob's reading of the transmitted measurement result, and his unitary correction. Provided that his frame configuration remains approximately constant between these steps, the  effective channel~\eqref{eq:maintheorem} is unaffected by arbitrary changes in reference frame alignment throughout the rest of the procedure.
\end{proof}

\subsubsection{Perfect scheme}

\ignore{
In general Procedure~\ref{proc:meastransferproc} communicates an infinite amount of reference frame information. 
\begin{proposition}\label{prop:rfinfotransfer}
Suppose Alice measures $i \in I_k$, performs Procedure~\ref{proc:meastransferproc}, and Bob receives $y \in C$. Bob now knows that the reference frame misalignment $g_{AB} \in G$ lies in the subset \begin{equation}\label{eq:subsetpossiblerfmisalsafterproc}\{g \in G \; | \; g^{-1} \cdot y = x \text{\emph{  for some }} x \in \sqcup_j E^k_j\}.\end{equation}
\end{proposition}
}

\begin{procedure}[Perfect scheme]\label{proc:perfectmeastransferproc}

Let $C$ be an unspeakable channel for $G$ (and therefore also for $H$), and let $(D^k_i,E^k_i)_{i \in I}$ be encoding schemes for $I_k$ compatible with $\sigma_k: I_k \times H \to I_k$, and where $E^k_i=X^k_i$, where $X^k_i \subset D^k_i$ is a finite set of readings in $C$, and moreover $H$ acts transitively on $\sqcup_i X^k_i$.

Alice measures in the basis $\{\ket{\phi_i}\}_{i \in I}$ \eqref{eq:phimeas} as in a standard teleportation protocol and obtains the result $i \in I_k$. The result is transmitted as follows.
\begin{enumerate}
\item Alice transmits the orbit label $k$ through a speakable channel.
\item Alice sends a reading $x^k_i \in X_i^k$ chosen uniformly at random.
\item Bob receives $y= g \cdot x^k_i \in g \cdot X_i^k = X_j^k \subset D_j^k$ and performs the correction $\rho( r_j(y) ) U_j \rho( r_j(y))^{\dagger}$, where $r_j(y) \in G$ is any element such that $r_j(y) \cdot x^k_j = y$ for some $x^k_j \in X^k_j$.
\end{enumerate}
In words, Bob realigns his frame (actively or passively) so that the reading he receives is $x^k_j \in X_j^k$, and then performs the correction $U_j$. Here $g$ is the reference frame transformation taking Alice's frame at the time of measurement onto Bob's frame at the time of receipt.
\end{procedure}
\noindent
\begin{proposition}[Effective channel for Procedure~\ref{proc:perfectmeastransferproc}]\label{prop:effchanrelaxedwrecording}
Suppose that Alice measures some result $i \in I_k$. Then the quantum channel induced by Procedure~\ref{proc:perfectmeastransferproc} is as follows:
\begin{equation}\label{eq:pointencodingchannelwithrecording}
\mathcal{T}_i(\sigma)
= \int\limits_{\text{\emph{Stab}}_G(x_i)} \emph ds \,[\rho( s)^{\dagger} U_i \rho(s)U_i^{\dagger}] \;(\sigma)
\end{equation}
Here $\emph{d}s$ is the Haar measure on $\emph{Stab}_G(x_i^k)$.
\end{proposition}
\begin{proof}
Alice measures $i \in I_k$ and communicates $x_i^k$ to Bob, who receives $y \in D_j$, where $y = g \cdot x_i^k = (r_j(y) h_{ij} s) \cdot x^k_i$ for $h_{ij} \in H$ such that $h_{ij} \cdot x^k_i = x^k_j$ (this exists because $H$ acts transitively on $\sqcup_i X^k_i$) and some $s \in \text{Stab}_G(x_i^k)$.

The distribution over $\text{Stab}_G(x_i^k)$ is uniform.  We therefore have the following expression for the effective channel:
\begin{align*}
\mathcal{T}_k(\rho) &= \int\limits_{\text{Stab}_G(x_i^k)} \text ds\, [\rho(r_j(y) h_{ij} s)^{\dagger} \rho( r_j(y) ) U_j \rho( r_j(y))^{\dagger} \rho(r_j(y) h_{ij} s)U_i^{\dagger}] \;(\sigma) \\
&= \int\limits_{\text{Stab}_G(x_i^k)} \text ds \,[\rho( h_{ij} s)^{\dagger} U_j \rho( h_{ij} s)U_i^{\dagger}] \;(\sigma) \\
&= \int\limits_{\text{Stab}_G(x_i^k)} \text ds \,[\rho( s)^{\dagger} U_i \rho(s)U_i^{\dagger}] \;(\sigma)
\end{align*}
At each step, we used the fact $\rho$ is a representation. 
For the final equality, we used equivariance of the unitary error basis.
\end{proof}
\noindent
In particular, this produces perfect teleportation for matched channels.
\begin{proposition}\label{prop:perfecttel}
Procedure~\ref{proc:perfectmeastransferproc} with the perfect encoding scheme on a matched channel (Definition~\ref{def:rfencodingschemes}) results in perfect teleportation.
\end{proposition}
\begin{proof}
The stabiliser of any reading is trivial, since the action is free.\end{proof}
\noindent
The perfect scheme also possesses the (ME) and (DR) properties, for exactly the same reasons as the tight scheme.

\infinitepraacknowledgements{\version}
\biblstyle{\version}
\bibliography{FiniteGroupTeleportation}

\appendix
\appendixpage

\section{Reference frame transformation rules}
\label{sec:refframetransfrulesproof}

In this appendix we briefly summarise the effect of reference frame transformations on measurements and operations. Let $\mathcal{F}$ be the space of reference frame configurations.
Let $V$ be the $d$-dimensional Hilbert space of a system whose states are described according to a reference frame.  The Hilbert space carries a unitary representation $\rho: G \to B(V)$, which encodes how states transform upon a change of reference frame: a state with vector $\ket{\psi}$ in reference frame $f\in \mathcal{F}$ will  have vector $\rho(g) \ket{\psi}$ in reference frame $g\cdot f$. Let $g_{AB} \in G$ be the reference frame transformation taking Alice's frame $f_A \in \mathcal{F}$ onto Bob's frame $f_B \in \mathcal{F}$; that is, $f_B = g_{AB} \cdot f_A$. We then have the following expressions:

\begin{proposition}\label{prop:framechangestatesandops}
A state with vector  $\ket{\psi}$ in Bob's frame has vector  $\rho(g)^{\dagger}\ket{\psi}$ in Alice's frame. An linear map with matrix  $M: V \to V$ in Bob's frame has matrix  $\rho(g)^{\dagger}M \rho(g)$ in Alice's frame. A general operation $\Phi: L(V) \to L(V)$ in Bob's frame is the operation  $[\rho(g)^{\dagger}] \circ \Phi \circ [\rho(g)]$ in Alice's frame.
\end{proposition} 

\begin{proof}
By definition a state described in Alice's frame as $\ket{\psi}$ will be described in Bob's frame as $\rho(g) \ket{\psi}$; the first equation follows immediately. 

For the linear maps, consider that an linear map is defined by its matrix elements in some orthonormal basis. Bob performs the operation with matrix elements $M_{ij}$ in his frame; that is, he performs the operation $M_B$ such that $\bra{i_B}M_B\ket{j_B} = M_{ij}$. Now note that $\ket{i_B} = \rho(g)^{\dagger} \ket{i_A}$, so $M_{ij}=\bra{i_B}M_B\ket{j_B} = \bra{i_A} \rho(g) M_B \rho(g)^{\dagger} \ket{j_A}$. In Alice's frame, therefore, Bob has performed the operation $M_B$ such that  $\rho(g) M_B \rho(g)^{\dagger} = M_A$; this operation is therefore related to $M_A$ by $M_A =  \rho(g)^{\dagger}M_B \rho(g)$.  The same argument can be extended to general operations by considering the Kraus maps.
\end{proof}
\section{Proof of Theorem~\ref{thm:maintheorem}}\label{app:effchanexpproof}
We now provide the postponed proof of this theorem.
\begin{theorem}[Effective channel for a general encoding scheme]
Suppose that Alice measures some result $i \in I_k$, where $D_k^i = E_k^i$ for all $i \in I_k$. Then the channel induced by Procedure~\ref{proc:meastransferproc} is as follows:
\begin{equation}
\mathcal{T}_k(\rho)= \frac{|I_k|}{\mu_C(E_0^k)} [\pi(c_i)] \circ \int\limits_G \left( \emph{d}g\, p(g)\, [\pi(g)^{\dagger} U_0  \pi(g)U_0^{\dagger}] \right) \circ [\pi(c_i)^{\dagger}] \;(\rho)
\end{equation}
Here $0 \in I_k$ is any element of the orbit;  the normalising factor $\mu_C(E_0^k)$ is the measure of $E_0^k$ in $C$; $p(g) = \int_{E_0^k \subset C} \emph{d}x \, \mathbbm{1}_{D_0^k}(g \cdot x)$, where $\mathbbm{1}_{D_0^k}$ is a continuous approximation to the indicator function for $D_0^k \subset C$; and $\{c_i\}_{i \in I_k}$, $c_i \in H$ are such that $c_i \cdot E^k_0 =E^k_i$.
\end{theorem}
\begin{proof}
We define $U(x) = U_j \;| \;x \in D^k_j$. Then, in Alice's frame, Bob's correction will be:
$$\pi(g_{AB})^{\dagger} U(g_{AB} \cdot x ) \pi(g_{AB}),$$
where $x \in E^k_i$ is the direction sent by Alice. Since both $g_{AB} \in G$ and $x \in E^k_i$ are unknown and uniformly distributed, we must average over both. When Alice measures $i \in I_k$, the channel is as follows for input state $\sigma$:
\begin{equation}
\mathcal{T}_i^{k}(\sigma)=\frac{1}{\mu_C(E_i^k)} \int\limits_{G \times C} \text dg \, \text dx \, \mathbbm{1}_{E^k_i}(x)\, [\rho(g)^{\dagger} U(g \cdot x) \rho(g) U_i^{\dagger}]\;(\sigma)
\end{equation}
Here $\mathbbm{1}_{E^k_i}$ is a continuous approximation to the indicator function for the region $E_i \subset C$.

First we show that $\mathcal{T}_i^k = [\rho(c_i)] \circ \mathcal{T}_0^k \circ [\rho(c_i)^{\dagger}]$; that is, every measurement result in a given orbit produces a similar channel. Indeed, since the product measure $dg\,d\phi$ is invariant under the left $G$-action $g_1 \cdot (g_2,x) = (g_2 g_1^{-1},g_1 \cdot x)$ on $G \times C$,  we can make the change of variables $(g,x) \mapsto (g c_i^{-1}, c_i \cdot x)$:
\begin{align*}
\mathcal{T}_i^{k}(\sigma)&= \frac{1}{\mu_C(E_i^k)} \int\limits_{G \times C} \text dg \, \text dx\, \mathbbm{1}_{E^k_i}(c_i \cdot x) [\rho(c_i)  \rho(g)^{\dagger} U(g \cdot x) \rho(g) \rho(c_i)^{\dagger} U_i^{\dagger} \rho(c_i) \rho(c_i)^{\dagger}]\, (\sigma) \\
&= \frac{1}{\mu_C(E_0^k)}\, [\rho(c_i)] \circ  \int\limits_{G \times C} \text dg\, \text dx\, \mathbbm{1}_{E^k_0}(x)  [\rho(g)^{\dagger} U(g \cdot x) \rho(g) U_1^{\dagger}] \; \circ [\rho(c_i)^{\dagger}] \, (\sigma)\\
&= [\rho(c_i)] \circ \mathcal{T}_0^k \circ [\rho(c_i)^{\dagger}]
\end{align*}
To obtain the first equality we changed variables and used the fact that $\rho$ is a representation. For the second equality we used $\mathbbm{1}_{E^k_i}(c_i \cdot x) = \mathbbm{1}_{E^k_0}$, linearity, and the fact that the action of $G$ on $C$ is measure-preserving. 
We can therefore restrict our attention to the channel where Alice measures the index $0 \in I_k$. 

We will now express the integral for the channel $T_0^{k}$ as a sum over integrals where Bob performs a definite correction. The action $\nu: (g,x) \mapsto g \cdot x$ is continuous; it follows that the preimages of the open sets $D^k_i$ under $\nu$ are open and therefore measurable. That the open sets $\nu^{-1}(D^k_i)$ cover $G \times C$ up to a set of measure zero follows immediately from the fact that the $D^k_i$ cover $C$ up to a set of measure zero and $\nu$ is a submersion. \ignore{Firstly, since  $G \backslash (\sqcup_i R_i)$ is a set of measure zero, for every $p \in G \backslash (\sqcup_i R_i)$ there is no open set $U \subset G$ containing $p$ such that $U \subset G \backslash (\sqcup_i R_i)$; if there were, then the Haar measure of this set would be nonzero, since the Haar measure is nonzero on every nonempty open set. Therefore for each point $p \in G \backslash (\sqcup_i R_i)$, there must be some tangent vector $v \in T_pG$ such that movement along $v$ takes one out of $G \backslash (\sqcup_i R_i)$. We know that $\nu$ is a submersion; so there must be a vector in the preimage of $v$ under the differential $D\nu$ along which parallel transport takes us out of $\nu^{-1} (G \backslash (\sqcup_i R_i))$. Therefore there is no open set in $G\times G$ containing the preimage of this point and which is made up entirely of points in $\nu^{-1} (G \backslash (\sqcup_i R_i))$. Since this is true of all points in $G \backslash (\sqcup_i R_i)$, the set is of measure zero.} We may therefore split the domain of integration over the $\nu^{-1}(D^k_i)$:
\begin{align*}
\mathcal{T}_0^k(\sigma) =  \frac{1}{\mu_C(E_0^k)} \sum_{i \in I_k} \int_{G \times C} \text dg \, \text dx \, \mathbbm{1}_{E^k_0}(x) \mathbbm{1}_{D^k_i}(g \cdot x)\,[
 \rho(g)^{\dagger} U_i \rho(g) U_0^{\dagger} ]\;(\sigma)
\end{align*}
Now we observe that the integrals over $\nu^{-1}(D_i^k)$ are identical for all $i \in I_k$:
\begin{align*}
\mathcal{T}_0^k (\sigma) &= \frac{1}{\mu_C(E_0^k)} \sum_{i\in I_k} \int_{G \times C} \text dg\, \text dx \, \mathbbm{1}_{E^k_0}(x)\mathbbm{1}_{D^k_i}(g \cdot x)\,[ \rho(c_i^{-1} g)^{\dagger} U_0 \pi(c_i^{-1} g) U_0]\;(\sigma) \\
&=\frac{|I_k|}{\mu_C(E_0^k)} \int\limits_{G \times C} \text dg\, \text dx \,\mathbbm{1}_{E^k_0}(x)\mathbbm{1}_{D^k_0}(g \cdot x)\,[ \rho(g)^{\dagger} U_0 \rho(g) U_0]\;(\sigma)
\end{align*}
The first equality uses that $U_i = \rho(c_i) U_0 \rho(c_i)^{\dagger}$; in the second we performed the change of variables $(g, x) \mapsto (c_i g, x)$ and noted that $\mathbbm{1}_{D^k_i}((c_i g) \cdot x) = \mathbbm{1}_{D^k_0}(g \cdot x)$, since $D_i^k = E_i^k$ for all $i,k$. By Fubini's theorem this may be evaluated as an iterated integral, where $x$ is integrated over first:
\begin{equation*}
\mathcal{T}_0^k(\sigma) =  \frac{|I_k|}{\mu_C(E_0^k)} \int\limits_{G} \text dg \int\limits_{C} \text dx\, \mathbbm{1}_{E^k_0}(x)\mathbbm{1}_{D^k_0}(g \cdot x)\, [\rho(g)^{\dagger} U_0  \rho(g)U_0]\;(\sigma)
\end{equation*}
This produces a weighting for $g \in G$ which is precisely the measure in $C$ of the set $D^k_0 \cap (g \cdot E^k_0)$. The result follows.
\end{proof}

\section{Calculations}
\label{sec:numerics}

In this section we derive the numerical results presented in Table~\ref{tbl:numericsintro}.

\subsection{Map purity and its calculation}\label{sec:chanqualmeasures}
\def\T{\mathcal{T}}

The measure we use to evaluate the success of the protocol is the \emph{map purity}~\cite{Roga2013,Roga2011,Ziman2008}. Recall that the \emph{Choi-Jamio\l{}kowski (CJ)}  state $\rho_{\T}$ of a channel $\T$ on a Hilbert space of dimension $d$ is  
$$\rho_{\T} = \frac{1}{2}(\mathbbm{1} \otimes \T) \, (\omega), $$
where $\omega$ is the density matrix of the state $\frac{1}{\sqrt{d}} \sum_{i=0}^{d-1}\ket{i}\otimes \ket{i}$. (For calculations, recall that the density matrix of the CJ state can be obtained by `reshuffling' the entries of the superoperator matrix of the channel~\cite{Roga2013}.)
\begin{definition}
The \emph{map purity}  $P(\T)$ of a channel $\T$ on a Hilbert space of dimension $d$ is  the normalised purity of its CJ state; that is,
\begin{equation}\label{eq:normalisedmappurity}P(\T) := 1 - \frac{S(\rho_T)}{\ln(d^2)} = 1+ \frac{\Tr(\rho_{\T}\ln(\rho_{\T}))}{\ln(d^2)}
\end{equation}
\end{definition}
\noindent
For numerical optimisation we will additionally use the linear map purity.
\begin{definition}
The \emph{linear map purity}  $P^L(\T)$ of a channel $\T$ on a Hilbert space of dimension $d$ is defined as the linear purity of its CJ state; that is,
$$P^{\text{L}}(\T) = \Tr(\rho_{\T}^2).$$
\end{definition}
\ignore{
$$P^{\text{L}}(\rho_\T) = \Tr(\rho_{\T}^2)$$

We will briefly discuss why we have chosen this particular measure, rather than, for instance, the average \textit{fidelity} over pure state inputs \cite{Bouda2009,Emerson2005,Marzolino2015}. Firstly, the purity cannot be increased by performing local unitary operations, unlike the fidelity. For instance, if a qubit channel performs a $X$ gate, the fidelity will be low; however, Alice can simply perform a $X$ gate before performing the protocol, yielding the identity channel with perfect fidelity. We assume Alice and Bob have full knowledge of the effect of reference frame uncertainty on the channel, so it sensible to choose a measure which excludes such cases.
}
\noindent 
The map purity in the qubit case, which we consider in our examples, is very similar to minimum purity over pure state inputs~\cite{Roga2011}.

By~\eqref{generalnaivechanneleqn}, the channels we consider are of the following sort.
\begin{definition}
A \emph{random unitary channel} is a channel of the form 
$$ \sigma \mapsto \int_X \d x [U(x)]\,(\sigma)$$
for some label space and probability measure $(X, \d x)$, where each $U(x)$ is a unitary matrix. 
\end{definition}
\noindent In particular, our random unitary channels are 
$$ \sigma \mapsto \sum_{i} \int_G \d g \; p(i)q(g)[U(i,g)] (\sigma),$$
where $U(i,g)$ are the unitaries, the label space is $I \times G$, and the probability measure on the label space is $p(i) \d g$; this is the product of the probability $p(i)$ of measurement result $i$ (which is uniform), and the Haar measure $\d g$ over the group $G$ of reference frame misalignments. A little straightforward algebra yields the following useful expression for the linear map purity of these channels.
\begin{proposition}[Linear map purity of a random unitary channel]\label{randomunitarypuritylemma}
Let $\T$ be a random unitary channel on a Hilbert space of dimension $d$. Let the random unitaries be indexed by a discrete index $I=\{0, \dots, n-1\}$ with probability distribution $p(i)$ and a a continuous index $g \in G$ with  probability measure $\d g$. Then:
\begin{equation}\label{eq:randomunitarypurityequation}
P^L(\T) = \frac{1}{d^2} \sum_{i,j=0}^{n-1} \int_{G \times G} p(i) p(j) \d g \d g' |\Tr(U(i,g)^{\dagger}U(j,g'))|^{2}
\end{equation}
\end{proposition}

\noindent
We now consider teleportation of quantum systems carrying fundamental representations of the reference frame transformation groups $\U(1)$ and $\SU(2)$. For each of these groups, we first find the UEB which optimises the linear map purity of the quantum channel resulting from a conventional protocol~\eqref{generalnaivechanneleqn}, and then calculate the map purity of the quantum channel arising from that UEB, obtaining the numbers in the second column of Table~\ref{tbl:numericsintro}. We then calculate the map purity for certain of our tight schemes, obtaining the numbers in the third column of that table.

\subsection{Calculations for $\U(1)$}\label{sec:u1calculations}
Here we consider the case  $G=\U(1)$, where the group of reference frame transformations acts on the qubit state as follows:
\begin{equation}
\label{u1representationdefn}
\theta \mapsto \begin{pmatrix}1 & 0 \\0 & e^{i \theta}\end{pmatrix}
\end{equation}

\paragraph{Conventional scheme.}
We begin by finding the UEB which optimises linear map purity for a conventional protocol. A general qubit UEB may be expressed as $U \mathcal{E} V$, where $U, V$ are arbitrary unitary matrices and $\mathcal{E} = \{ X_0,X_1,X_2,X_3 \}$ is the Pauli UEB~\eqref{eq:paulis}. \ignore{By the correspondence between UEBs and tight teleportation protocols (Theorem~\ref{Wernertighttheorem}) we see that with only speakable communication and uniform reference frame uncertainty, the channel will be random unitary with equiprobable unitaries.}
Since we ignore global phase, we need only consider unitaries up to their induced rotation of the Bloch sphere. Let $R_{\hat{n}}(\theta)$ be a Bloch sphere rotation through an angle $\theta$ around the $\hat{x}$ axis; let $X_i$ be a Pauli rotation (that is, a rotation through an angle $\pi$ around the $x$-, $y$- or $z$ axis); and let $\hat{x},\hat{y}$ be two unit vectors which correspond to the choice of~UEB. Then the  equiprobable unitaries are as follows:
\begin{align}
U_{ig}  &= g V^{\dagger} X_i U^{\dagger} g^{\dagger} U X_i V \label{preshufflerandunitaryu1}\\
                &\sim V g V^{\dagger} X_i U^{\dagger} g^{\dagger} U X_i\label{postshufflerandunitaryu1}\\
                &= R_{\hat{x}}(\theta)R_{X_i(\hat{y})}(-\theta)
\end{align}
We write $\sim$ to indicate that replacing unitaries (\ref{preshufflerandunitaryu1}) with unitaries (\ref{postshufflerandunitaryu1}) will yield a channel with the same purity, because of cyclicity of the trace in~\eqref{eq:randomunitarypurityequation}. The second equality follows by the fact that conjugating a rotation $R_{\hat{x}}(\theta)$ by another rotation $Q$ gives $Q R_{\hat{x}}(\theta) Q^{-1} = R_{Q(\hat{x})}(\theta)$. By Lemma \ref{randomunitarypuritylemma} we therefore have the following expression for the effective channel:
\begin{equation}\label{speakablepurityforu(1)}
P(\T) = \frac{1}{256 \pi ^2} \sum_{i,j} \int_{0}^{2\pi} \int_{0}^{2\pi} \d\theta_1 \d\theta_2 |\Tr[R_{X_j(\hat{y})}(-\theta_2)R_{X_i(\hat{y})}(\theta_1)R_{\hat{x}}(\theta_2-\theta_1)]|^{2}
\end{equation}
Here the choice of UEB corresponds to a choice of two unit vectors $(\hat{x},\hat{y})$ or equivalently a choice of angles $(\psi_{\hat{x}},\psi_{\hat{y}}, \phi_{\hat{x}},\phi_{\hat{y}}) \in [0,\pi]^2 \times [0,2\pi]^2$. The factor in front of the integral is a product of the normalisation factors for the parameterisation of $\U(1)$ and the $1/4$ probabilities for measurement results $i$ and $j$. The simplicity of the integral allows us to numerically evaluate it for given $\hat{x},\hat{y}$ with negligible error. We performed Nelder-Mead maximisation over $\hat{x}$, $\hat{y}$ and found optimality of the Pauli UEB, corresponding to angles $(0,0,0,0)$.
The normalised map purity for this UEB is
$$1+ \frac{1}{\ln(4)}(0.75\ln(0.75)+0.25\ln(0.25)) \simeq 0.59.$$

\paragraph{Tight scheme.} 
We must choose a finite subgroup $H \subset  \U(1)$ for which an equivariant UEB exists. In~\cite{Verdon2018} the largest such subgroup was shown to be $H \simeq \mathbb{Z}_4$, with a two-parameter family of equivariant UEBs:
\begin{align*}
U_0 =  R_{\hat{z}}(\theta-\pi)
&&U_1 = R_{\hat{z}}(\phi) X R_{\hat{z}}(-\phi)
&&U_2 =  R_{\hat{z}}(\phi) Y R_{\hat{z}}(-\phi)
&&U_3 =R_{\hat{z}}(\theta)
\end{align*}
The Pauli UEB is the member of this family with parameters $\theta = \pi, \phi = 0$.  The tight reference frame encoding scheme for this family of UEBs was given in Figure~\ref{fig:examplemeasregions}.

We use Theorem~\ref{thm:maintheorem} to calculate the superoperator for the effective channel. Because the group is abelian, conjugation by $\pi(c_i)$ is irrelevant, so the channel will be identical for measurements 1 and 2. 
For a similar reason we need only consider the Pauli UEB, since all UEBs in the family yield identical channels. It is easy to derive an analytic expression for $p(\theta)$:
\begin{equation}
p(\theta) = \abs*{\frac{(\theta+\pi/2)}{\pi} - \floor*{\frac{1}{2} + \frac{(\theta+\pi/2)}{\pi}}}
\end{equation}
The effective channel when Alice measures 1 is
\begin{equation}
4 \int_{-\pi}^{\pi} \text d\theta\, p(\theta)\,[\rho(\theta)^{\dagger} U_1 \rho(\theta) U_1^{\dagger}] (\sigma),
\end{equation}
and the channel for result 2 is similar. Since measurement results 0 and 3 yield perfect teleportation, we obtain the action~\eqref{eq:u1dmtwirl} of the effective channel on input density matrices. The normalised map purity for this effective channel is $$1+ 0.5 \ln(0.5) \simeq 0.65.$$

\subsection{Calculations for $\SU(2)$}
\label{sec:su2calculations}
We now consider the case  $G=\SU(2)$, acting on a qubit state by its defining representation.

\paragraph{Conventional scheme.}
We have a channel of the form \eqref{eq:randomunitarypurityequation}, which involves integration over $\SU(2)$. In order to obtain a parametrisation and measure for the integral, we use the isomorphism between $\SU(2)$ and the unit quaternions. These quaternions, being diffeomorphic to the 3-sphere $S^3$, may be parametrised by hyperspherical coordinates $(\theta,\psi,\phi) \in D$, where $D=[0,\pi] \times [0, \pi] \times [0,2\pi]$. This parametrisation is inherited by $\SU(2)$, along with the Haar measure $\text d\Omega$ on $S^3$, as follows: 
\begin{align*}
g(\theta,\psi,\phi) &= \begin{pmatrix}
\cos(\theta) + i \sin(\theta)\sin(\psi)\sin(\phi) & \big(\! \cos(\psi)+i\cos(\phi)\sin(\psi)\big)\sin(\theta) \\
- \big( \!\cos(\psi) - i \cos(\phi)\sin(\psi) \big) \sin(\theta) & \cos(\theta) - i \sin(\phi)\sin(\psi)\sin(\theta)
\end{pmatrix}
\\
\text  d\Omega &= \frac{1}{2\pi^2}\sin^2(\theta) \sin(\psi) \, \d\theta \d\psi \d\phi
\end{align*} 

We consider the integrand. Expanding the UEB elements in the form $U \mathcal{E} V$, where $U, V$ are arbitrary unitary matrices and $\mathcal{E} = \{ X_0,X_1,X_2,X_3 \}$ is the Pauli UEB, we see that the unitaries of the channel will be, for all $Y \in \SU(2)$ and $i \in I=\{1, \dots, 4 \}$,
\begin{align*}
U(Y, i) &= Y V^{\dagger} X_i^{\dagger} U^{\dagger} Y^{\dagger}  U X_i V \\
                &\sim V Y V^{\dagger}  X_i^{\dagger} U^{\dagger} Y^{\dagger}  U X_i
\end{align*}
where the equivalence is again a consequence of the cyclicity of the trace in \eqref{eq:randomunitarypurityequation}. We therefore obtain the following equation for the map purity:
\begin{equation}\label{speakablepurityforsu(2)}
P(\T) = \frac{1}{32} \int_{D \times D} \d\Omega_1 \d\Omega_2 |\Tr[ X_i Y_1 X_i\widetilde{U} Y_1^{\dagger} Y_2 \widetilde{U}^{\dagger} X_j Y_2^{\dagger} X_j ]|^{2}
\end{equation}
Here we performed a change of variables from $Y_i$ to $\widetilde{Y}_i = VY_iV^{\dagger}$, using the invariance of the measure; we omit the tilde on the new variable. We also wrote $\widetilde{U} := VU$; note that this is the the only significant element in our choice of UEB.

There are only  three relevant angle variables in the choice of UEB, corresponding to a choice of a single unitary $\widetilde{U} := VU$. We performed random sampling of 100 angle triples; none of these UEBs outperformed the Pauli matrices, whose normalised map purity is $$1-\frac{1}{2\ln(4)}\left(\ln\left(\frac{1}{2}\right)+\ln\left(\frac{1}{6}\right)\right) \simeq 0.21.$$

\paragraph{Tight scheme with rod channel.} The action on the rod channel considered in Section~\ref{sec:example} can be most easily expressed using the inner product--preserving isomorphism of $\SU(2)$-spaces 
\begin{equation}
\begin{split}
S^2 \subset \mathbb{R}^3 &\overset{\alpha}{\to} B(\mathbb{C}^2) \\
(n_x, n_y, n_z) &\mapsto \frac{I + (n_x,n_y,n_z)\cdot (X,Y,Z)}{2},
\end{split} 
\end{equation}
where $I, X, Y$ and $Z$ are the Pauli matrices, $S^2$ carries the obvious quotient left action of $\SU(2)$, and $B(\mathbb{C}^2)$ carries the left action of $\SU(2)$ by conjugation. The encoding and decoding regions are then made up of Voronoi cells for the cardinal points under the metric derived from the Hilbert-Schmidt inner product.

Using the above identification, we calculated $p(g)$ and the integral~\eqref{eq:integralfromexample} using Monte Carlo integration with rejection sampling~\cite{Reiher1966}, took the average over the four measurement results, and found normalised map purity $0.44 \pm 0.03$.

\paragraph{Tight scheme with reference frame channel.} 

Again, we choose the largest possible subgroup $H \subset  \SU(2)$ for which an equivariant UEB exists; in previous work~\cite{Verdon2018} this was shown to be $H \simeq \BOct$, where $\BOct$ is the binary octahedral group, which has order 48. The Pauli UEB is, up to phase, the unique UEB equivariant for this subgroup.

We choose the encoding and decoding regions to be Voronoi cells for the elements of $\BOct < SU(2)$ under the Frobenius distance function. 

\ignore{
We show in Appendix~\ref{sec:Voronoicells} that the Frobenius distance function (Definition \ref{ex:frobdistancefunction}) is an invariant distance function for $\SU(2)$ such that the Voronoi cell of the identity of any subgroup in $\SU(2)$ is a fundamental domain (Proposition \ref{prop:frobidvcellisfunddomain}). Let $F$ be the Voronoi cell of the identity element of the subgroup $H= \BOct$. 

The channel is perfect for measurement result 0. However, $\{U_1,U_2, U_3\}$ is a 3-orbit up to a phase under the conjugation action, isomorphic as a right $H$-set to the right coset space obtained by taking the quotient of $H$ by a certain subgroup~$K= \text{Stab}(U_1) \subset  H$.
\ignore{$$K=\text{Stab}(U_1) = \{I, -I, iX,i Y, iZ, -iX, -iY, -iZ\}.$$} We choose right coset representatives of $K$ in $H$ as follows: 
\begin{align*}
c_1 =
\begin{pmatrix}
1 & 0 \\
0 & 1 
\end{pmatrix} 
&&c_2 =
\frac{1}{2}\begin{pmatrix}
 1-i & -1-i \\
 1-i & 1+i \\
\end{pmatrix} 
&& c_3 = \frac{1}{\sqrt{2}}
\begin{pmatrix}
 1+i & 0 \\
 0 & 1-i \\
\end{pmatrix}
\end{align*}
The channel expression is given by Theorem~.
}
We evaluated the integral in Theorem~\ref{thm:maintheorem} using Monte Carlo integration with rejection sampling, took the average over the four measurement results, and calculated the normalised map purity of the effective channel as $0.32 \pm 0.02$.

\end{document}